\documentclass[12pt, a4paper]{article}

\usepackage{subfigure}
\usepackage{epsfig}

\usepackage{graphicx}
\usepackage{latexsym}
\usepackage{amssymb,amsmath,bm,epsfig}

\newtheorem{thm}{Theorem}[section]

\newtheorem{prop}[thm]{Proposition}
\newtheorem{cor}[thm]{Corollary}

\newenvironment {proof} {\noindent {\em Proof:}}{\hspace*{\fill}$\Box$\par\vspace{4mm}}

\def\mc#1{\mathcal #1}
\def\mb#1{\mathbb #1}

\def\E{\mathop{\rm \raisebox{0.1pt} E}}

\def\corref#1{Corollary \ref{#1}}

\def\tblref#1{Table \ref{#1}}
\def\figref#1{Figure \ref{#1}}
\def\secref#1{Section \ref{#1}}
\def\subsecref#1{Subsection \ref{#1}}

\title{Buffered environmental contours}

\author{K. R. Dahl \& A. B. Huseby} 

\date{}

\begin{document}

\maketitle

\abstract{The main idea of this paper is to use the notion of buffered failure probability from probabilistic structural design, first introduced by \cite{RockafellarRoyset}, to introduce buffered environmental contours. Classical environmental contours are used in structural design in order to obtain upper bounds on the failure probabilities of a large class of designs. The purpose of buffered failure probabilities is the same. However, in constrast to classical environmental contours, this new concept does not just take into account failure vs. functioning, but also to which extent the system is failing. For example, this is relevant when considering the risk of flooding: We are not just interested in knowing whether a river has flooded. The damages caused by the flooding greatly depends on how much the water has risen above the standard level.}

\section{Introduction}
\label{sec: intro}
Environmental contours are widely used as a basis for e.g., ship design. Such contours allow the designer to verify that a given mechanical structure is safe, i.e, that the failure probability is below a certain value. A realistic model of the environmental loads and the resulting response is crucial for structural reliability analysis of mechanical constructions exposed to environmental forces. See \cite{WitUCBH93} and \cite{HW:ENvContLin09}. For applications of environmental contours in marine structural design, see e.g., \cite{bhoe:CombContHsT10}, \cite{folnpq:RelAnFPSOWestAfrica2013}, \cite{Jonathan:OMAE2011-49886}, \cite{Moan:DevAccLiStCri09} and \cite{DitlevsenStocModJointWaveWind02}.

The traditional approach to environmental contours is based on the well-known \emph{Rosenblatt transformation} introduced in \cite{Rosenblatt52}. This transformation maps the the environmental variables into independent standard normal variables. Using the transformed environmental variables a contour with the desired properties can easily be constructed by identifying a sphere centered in the origin and with a suitable radius. More specifically, the sphere can be chosen so that any non-overlapping convex failure region has a probability less than or equal to a desired exceedence probability. The corresponding environmental contour in the original space can then be found by transforming the sphere back into the original space.

Alternatively, an environmental contour can be constructed directly in the original space using Monte Carlo simulation. See \cite{HusebyVN-EnvCont-OE2013}, \cite{HusebyVN-EnvCont-SS2015} and \cite{HusebyVN-EnvCont-ESREL2014}. Contours constructed using this approach will always be convex sets. This yields a more straightforward interpretation of the contours. Another advantage of this approach is a more flexible framework for establishing environmental contours, which for example simplifies the inclusion of effects such as future projections of the wave climate related to climatic change. See \cite{VanemB-G:AOR2011}.

In the present paper we introduce a new concept called buffered environmental contours. This concept is based on the notion of buffered failure probability from probabilistic structural design, first introduced by \cite{RockafellarRoyset}. Contrary to classical environmental contours, this new concept does not just take into account failure vs. functioning, but also to which extent the system is failing. For example, this is relevant when considering the risk of flooding: We are not just interested in knowing whether a river has flooded. The damages caused by the flooding greatly depends on how much the water has risen above the standard level.

The structure of this paper is as follows: In \secref{sec: buffered}, we recall the classic definition of failure probability in probabilistic structural design and compare this to the concept of buffered failure probability, as defined in \cite{RockafellarRoyset}. Furthermore, we recall some of the arguments favoring the buffered failure probability over the regular failure probability. Then, in \secref{sec: environmental_contours}, we recall the concept of environmental contours and how such contours are used in structural design in order to find upper bounds on the failure probabilities of a large class of designs. In \secref{sec: buffered_environmental_contours}, we introduce the new concept of buffered environmental contours, and argue that these contours are better suited than the classical ones in cases where the level of malfunctioning is important. Finally, in \secref{sec: numerical_example}, we apply the proposed contours to a real life example, and compare the contours to the classical environmental contours.

\section{Structural design and the buffered failure probabity}
\label{sec: buffered}
In probabilistic structural design, it is common to define a \emph{performance function}\footnote{The performance function is sometimes called the \emph{limit-state function}.} $g(\bm{x}, \bm{V})$ depending on some design variables $\bm{x} = (x_1, x_2, \ldots, x_m)'$ and some environmental quantities\footnote{Environmental quantities should here be understood in a broad sense. E.g., for marine structures such quantities typically includes wave height and period. For other types of structures, one may consider e.g., material quality, effects of erosion or corrosion as environmental quantities.} $\bm{V} = (V_1, V_2, \ldots, V_n)' \in \mc{V}$, where $\mc{V} \subseteq \mb{R}^n$. The design variables can be influenced by the designer of the structure, and may respresent material type or layout. The quantities are usually random, and cannot be directly impacted by the designer. Hence, they may describe environmental conditions, material quality or loads. To emphasize the randomness of the quantities, we denote them by captial letters. In contrast, the design variables are controlled by the designer and hence denoted by small letters.

For a given design $\bm{x}$, $g(\bm{x}, \bm{V})$ represents the performance of the structure, and is called the \emph{state of the structure}. A given mechanical structure can withstand environmental stress up to a certain level. The \emph{failure region} of the structure is the set of states of the environmental variables that imply that the structure fails. The performance function is defined such that if $g(\bm{x}, \bm{V}) > 0$, the structure is \emph{failed}, while if $g(\bm{x}, \bm{V}) \leq 0$, the structure is \emph{functioning}. Moreover, for a given $\bm{x}$ the set $\mc{F}(\bm{x}) = \{\bm{v} \in \mc{V} : g(\bm{x}, \bm{v})> 0\}$ is called the \emph{failure region} of the structure\footnote{In some papers, such as \cite{HusebyVN-EnvCont-OE2013}, the failed states are defined as the states such that $g(\bm{x}, \bm{V}) < 0$. This is just a matter of choice of notation.}.

\subsection{The failure probability, reliability and approximation methods}
\label{subsec: failure_prob}
The failure probability, denoted by $p_f(\bm{x})$, of the structure is the probability that the structure is failed. That is, $p_f(\bm{x}) = P(g(\bm{x}, \bm{V}) > 0)$. If $f_{\bm{V}}(\bm{v})$ is the joint probability density function for the random vector $\bm{V}$, the failure probability is given by:
\begin{equation}
\label{eq: failure_prob}
p_f(\bm{x}) = \int_{\mc{F}(\bm{x})}  f_{\bm{V}}(\bm{v}) d\bm{v}.
\end{equation}


For a given $\bm{x}$ the \emph{reliability}, $R(\bm{x})$, of the system is defined as the probability that the system is functioning, i.e.:
\begin{equation}
\label{eq: reliability}
R(\bm{x}) = 1 - p_f(\bm{x})
\end{equation}

A classic problem is to compute the reliability of the system. In order to do so, we need to compute the integral \eqref{eq: failure_prob}. In many cases it is difficult to obtain and analytical solution to this. To overcome this issue various approximation methods have been proposed. Two traditional methods for doing this are the \emph{first-order reliability method} (FORM) and the \emph{second-order reliability method} (SORM). The basic idea of the first-order reliability method is to approximate the failure boundary at a spesific point by a first order Taylor expansion. The idea behind SORM is similar, but using a second order Taylor expansion instead. In both cases, the approximated failure probability can be used to optimize the structural design, i.e. determine a feasible design which has an acceptable failure probability.

\subsection{Return periods}
\label{sec: returnPeriods}
As is common in structural design models, we view $\bm{V}$ as representing the average value of the relevant environmental variables in a suitable time interval of length $L$. Based on this and knowledge of the performance function $g$ it is possible to compute the so-called \emph{return period}. This is done as follows:

We consider the environmental exposure of the given design from time $t \geq 0$. The time axis is divided into intervals of some specified length $L$, and we let $\bm{V}_i$ denote the average environmental quantity in the $i$th period, $i = 1, 2, \ldots$. It is common to assume that $\bm{V}_1, \bm{V}_2, \ldots$ are independent and identically distributed. This is a fairly strict assumption, but as it is so frequently used in structural design, we assume this as well. We then let $T := \min \{ i : g(\bm{x}, \bm{V}_i) >0\}$. By the assumptions it follows that $T$ is geometrically distributed with probability $p_f = P(g(\bm{x}, \bm{V}) >0)$. The \emph{return period} is defined as $E[T] = 1/p_f$. Thus, the return period can be interpreted as a property of the distribution of $g(\bm{x}, \bm{V})$. Hence, it suffices to analyze this distribution, which is what we will focus on in this paper.

\subsection{The buffered failure probability}
The approximations made by FORM and SORM can sometimes be too crude and ignore serious risks. Therefore, we will consider the buffered failure probability, introduced by \cite{RockafellarRoyset} as an alternative to the failure probability. This concept relates closely to the conditional value-at-risk (also called expected shortfall, average value-at-risk or expected tail loss), which is a notion frequently used in mathematical finance and financial engineering, see \cite{Pflug}, \cite{Rockafellar} as well as \cite{RockafellarUryasev}. 

Recall that for any level of probability $\alpha$, the $\alpha$-quantile of the distribution of a random variable is the value of the inverse of its cumulative distribution function at $\alpha$. For the random variable $g(\bm{x}, \bm{V})$, we let $q_{\alpha}(\bm{x})$ denote its $\alpha$-quantile. Similarly, for any probability level $\alpha$, the $\alpha$-superquantile of $g(\bm{x},\bm{V})$, $\bar{q}_{\alpha}(\bm{x})$, is defined as:
\begin{equation}
\label{eq: superquantile}
\bar{q}_{\alpha}(\bm{x}) =  E[g(\bm{x}, \bm{V}) | g(\bm{x}, \bm{V}) > q_{\alpha}(\bm{x})].
\end{equation}
That is, the $\alpha$-superquantile is the conditional expectation of $g(\bm{x}, \bm{V})$ when we know that its value is greater than or equal the $\alpha$-quantile. \cite{RockafellarRoyset} then define the buffered failure probability, $\bar{p}_f(\bm{x})$, as follows:
\begin{equation}
\label{eq: buffered_failure_prob}
\bar{p}_f(\bm{x}) = 1-\alpha,
\end{equation}
where $\alpha$ is chosen so that $\bar{q}_{\alpha}(\bm{x}) = 0$. Note that from the previous definitions we have: 
\begin{equation}
\label{eq: alt_buffered_failure_prob}
\bar{p}_f(\bm{x})= P(g(\bm{x}, \bm{V}) > q_{\alpha}(\bm{x})) = 1-F(q_{\alpha}(\bm{x}))
\end{equation}
where $F$ denotes the cumulative distribution function of $g(\bm{x},\bm{V})$. 

\begin{figure}[h]
\begin{center}
\includegraphics[width=0.5\textwidth]{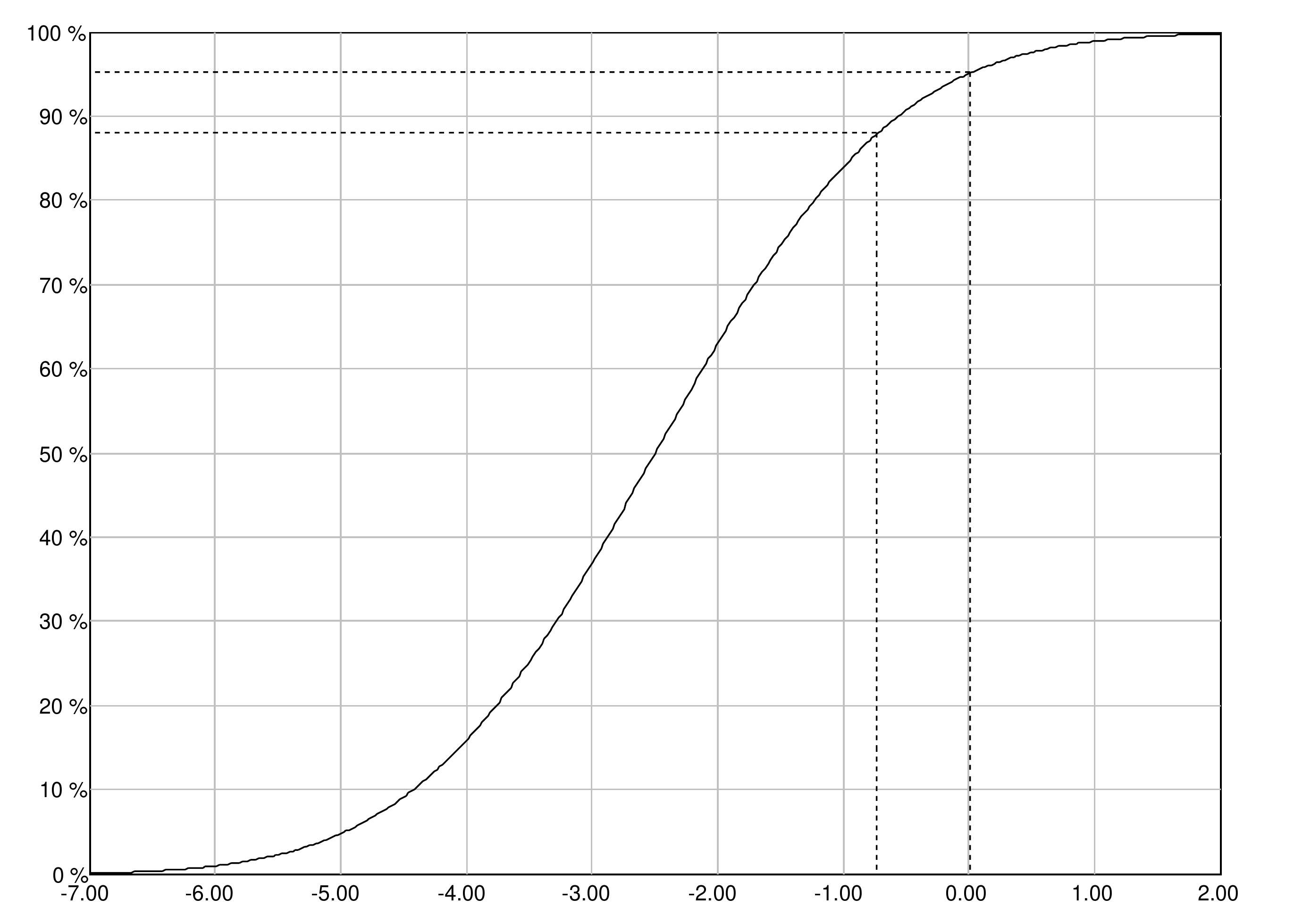}
\end{center}
\caption{Buffered failure probability calculation where: $p_f(\bm{x}) = 0.048$, $q_{\alpha}(\bm{x}) = -0.743$, $\alpha = F(q_{\alpha}(\bm{x})) = 0.879$, and $\bar{p}_f(\bm{x}) = 1 - \alpha = 0.121$.}
\label{fig: bufProb_gauss}
\end{figure}

In order to show how to calculate the buffered failure probability $\bar{p}_f(\bm{x})$, we consider the plot shown in \figref{fig: bufProb_gauss}. The curve in the plot represents the cumulative distribution function of the performance function, $g(\bm{x}, \bm{V})$. As an example we have chosen a Gaussian distribution with mean value $-2.5$ and standard deviation $1.5$. For this distribution we have $F(0) = 0.952$, as can also be seen in the figure by considering the right-most vertical dashed line starting at $0$ on the x-axis, and the corresponding upper horizontal dashed line starting at $0.952$. Hence, we get that $p_f(\bm{x}) = 1 - F(0) = 0.048$. In the figure $p_f(\bm{x})$ is the distance between $100\%$-line and the upper horizontal dashed line. 

Using e.g., Monte Carlo simulation it is easy to estimate $q_{\alpha}(\bm{x})$, and we find that $q_{\alpha}(\bm{x}) = -0.743$. In the figure $q_{\alpha}(\bm{x})$ is represented by the leftmost vertical dashed line. By following this line until it crosses the cumulative curve, we find that $\alpha = F(q_{\alpha}(\bm{x})) = 0.879$. Finally, the buffered failure probability is found to be $\bar{p}_f(\bm{x}) = 1 - \alpha = 0.121$. In the figure $\bar{p}_f(\bm{x})$ is the distance between $100\%$-line and the lower horizontal dashed line. 

It is easy to see that we always have $q_{\alpha}(\bm{x}) \leq 0$, and thus, it follows that $\alpha = F(q_{\alpha}(\bm{x})) \leq F(0)$. This implies that:
$$
\bar{p}_f(\bm{x}) = 1 - \alpha \geq 1 - F(0) = p_f(\bm{x}).
$$
Hence, it follows that the buffered failure probability is more conservative than the failure probability. See \cite{RockafellarRoyset} for a detailed discussion of this. 

\medskip

\cite{RockafellarRoyset} present several advantages of using the buffered failure probability instead of the regular failure probability. The following are some of the key arguments:

\begin{itemize}
\item{In general, the failure probability $p_f(\bm{x})$ cannot be computed analytically, and the techniques commonly used to approximate it, such as FORM or Monte Carlo methods, can sometimes ignore serious risks. This makes it problematic to apply standard non-linear optimization algorithms in connection to structure design. In contrast, non-linear optimization algorithms are directly applicable when using the buffered failure probability instead.}
\item{The buffered failure probability contains more information about the tail behaviour of the distribution of $g(\bm{x},\bm{V})$ than the failure probability.}
\item{The buffered failure probability can lead to more computational efficiency in design optimization when the performance function $g(\bm{x},\bm{V})$ is expensive to evaluate.}
\end{itemize}

The \emph{buffered reliability}, $\bar{R}(\bm{x})$, of the structure is defined as $\bar{R}(\bm{x}) = 1-\bar{p}_f(\bm{x})$. Since $p_f(\bm{x}) \leq \bar{p}_f(\bm{x})$, it follows that $R(\bm{x}) \geq \bar{R}(\bm{x})$. That is, the reliability of the system is greater than or equal to the buffered reliability. Again, this essentially says that the buffered reliability is more conservative than the reliability.

\section{Environmental contours}
\label{sec: environmental_contours}
Environmental contours are typically used during the early design phases where the exact shape of the failure region is typically \emph{unknown}. At this stage it it may not be possible to express a precise functional relationship between a set of design variables $\bm{x}$ and the performance of the structure. Instead we skip $\bm{x}$ in the notation and let the design options be embedded in the performance function $g(\bm{V})$ itself. In particular we denote the failure region simply by $\mc{F}$, while the corresponding failure probability, $P(\bm{V} \in \mc{F})$, is denoted by $p_f(\mc{F})$. 

Although $\mc{F}$ is unknown, it may still be possible to argue that $\mc{F}$ belongs to some known family, $\mc{E}$, of failure regions. As in the previous sections we consider cases where the environmental conditions can be described by a stochastic vector $\bm{V} \in \mb{R}^n$ with a known distribution. An important part of the probabilistic design process is then to make sure that $P(\bm{V} \in \mc{F})$ is acceptable for all $\mc{F} \in \mc{E}$.

In order to avoid failure regions with unacceptable probabilities, it is necessary to put some restrictions on the family $\mc{E}$. This is done by introducing a set $\mc{B} \subseteq \mb{R}^n$ chosen so that for any relevant failure region $\mc{F}$ which do not overlap with $\mc{B}$, the failure probability $P(\bm{V} \in \mc{F})$ is \emph{small}. The family $\mc{E}$ is chosen relative to $\mc{B}$ so that $\mc{F} \cap \mc{B} \subseteq \partial\mc{B}$ for all $\mc{F} \in \mc{E}$, where $\partial\mc{B}$ denotes the boundary of $\mc{B}$. This boundary is then referred to as an \emph{environmental contour}. See \figref{fig:env_contour}.

\begin{figure}[h!]
\begin{center}
\includegraphics[width=0.4\textwidth]{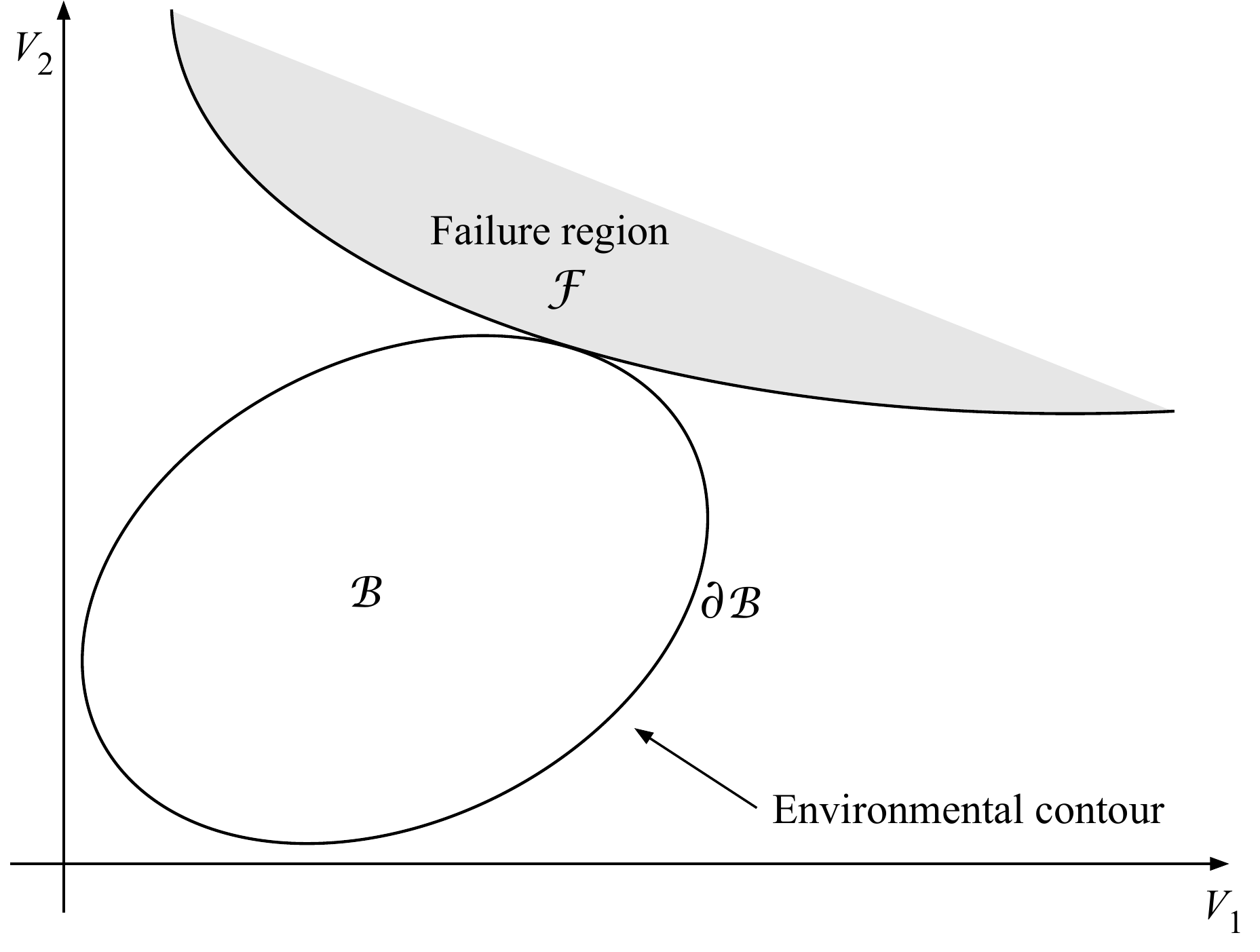}
\end{center}
\caption{An environmental contour $\partial \mc{B}$ and a failure region $\mc{F}$.}
\label{fig:env_contour}
\end{figure}

Following \cite{HusebyVE-EnvCont-ESREL2017} we define the \emph{exceedence probability} of $\mc{B}$ with respect to $\mc{E}$ as:
\begin{equation}
\label{eq:exceedenceProb}
P_{e}(\mc{B}, \mc{E}) := \sup \{p_f(\mc{F}) : \mc{F} \in \mc{E}\}.
\end{equation}
For a given \emph{target probability} $P_e$ the objective is to choose an environmental contour $\partial\mc{B}$ such that:
$$
P_{e}(\mc{B}, \mc{E}) = P_e
$$

We observe that the exceedence probability defined above represents an upper bound on the failure probability of the structure assuming that the true failure region is a member of the family $\mc{E}$. Of particular interest are cases where one can argue that the failure region of a structure is \emph{convex}. That is, cases where $\mc{E}$ is the class of all convex sets which do not intersect with the interior of $\mc{B}$. In the remaining part of the paper we will assume that $\mc{E}$ satisfies this. 

\subsection{Monte Carlo contours}
\label{sec: mc_contours}
There are many possible ways of constructing environmental contours. In this paper we focus on the Monte Carlo based approach first introduced in \cite{HusebyVN-EnvCont-OE2013}, and improved in \cite{HusebyVN-EnvCont-SS2015} and \cite{HusebyVN-EnvCont-ESREL2014}.

Let $\mc{U}$ be the set of all unit vectors in $\mb{R}^n$, and let $\bm{u} \in \mc{U}$. We then introduce a function $C(\bm{u})$ defined for all $\bm{u} \in \mc{U}$ as: 
\begin{equation}
\label{eq:Cfunction}
C(\bm{u}) := \inf\{C : P(\bm{u}'\bm{V} > C) \leq P_e\}
\end{equation}
Thus, $C(\bm{u})$ is the $(1-P_e)$-quantile of the distribution of $\bm{u}'\bm{V}$. Given the distribution of $\bm{V}$, the function $C(\bm{u})$  can easily be estimated by using Monte Carlo simulation. Thus, let $\bm{V}_1, \ldots, \bm{V}_N$ be a random sample from the distribution of $\bm{V}$. We then choose $\bm{u} \in \mc{U}$, and let $Y_r(\bm{u}) = \bm{u}'\bm{V}_r$, $r = 1, \ldots, N$. These results are sorted in ascending order:
$$
Y_{(1)} \leq Y_{(2)} \leq \cdots \leq Y_{(N)}
$$
Using the sorted numbers we first estimate $C(\bm{u})$. Since $C(\bm{u})$ is the $(1-P_e)$-quantile in the distribution, a natural estimator is:
$$
\hat{C}(\bm{u}) = Y_{(k)},
$$
where $k$ is determined so that:
$$
\frac{k}{N} \approx 1 - P_e.
$$
Note, however, that this estimator can be improved considerably by using importance sampling. See \cite{HusebyVN-EnvCont-ESREL2014} for details. 

\medskip

For each $\bm{u} \in \mc{U}$, we also introduce the halfspaces:
\begin{align*}
\Pi^-(\bm{u}) &= \{\bm{v} : \bm{u}' \bm{v} \leq C(\bm{u})\},\\
\Pi^+(\bm{u}) &= \{\bm{v} : \bm{u}' \bm{v} > C(\bm{u})\}.
\end{align*}
We then define the environmental contour as the boundary $\partial\mc{B}$ of the \emph{convex set} set $\mc{B}$ given by:
\begin{equation}
\label{eq:env_contourSet}
\mc{B} := \bigcap_{\bm{u} \in \mc{U}} \Pi^-(\bm{u})
\end{equation}
It follows that the exceedence probability of $\mc{B}$ with respect to $\mc{E}$ is given by:
\begin{align*}
P_e(\mc{B}, \mc{E}) 
&= \sup \{p_f(\mc{F}) : \mc{F} \in \mc{E}\} \\[2mm]
&= \sup \{p_f(\Pi^+(\bm{u})) : \bm{u} \in \mc{U}\} \\[2mm]
&= \sup_{\bm{u} \in \mc{U}} P(\bm{u}' \bm{V} > C(\bm{u})) = P_e,
\end{align*}
where the second equality follows since we have assumed that $\mc{F}$ is convex and hence contained in $\Pi^+(\bm{u})$ for all $\mc{F} \in \mc{E}$. In fact for all $\bm{u} \in \mc{U}$ we have $\Pi^+(\bm{u}) \in \mc{E}$ as well, and these halfspaces are the maximal sets within $\mc{E}$. Moreover, the last equation follows by the definition of $C(\bm{a})$ given in \eqref{eq:Cfunction}. Thus, we conclude that the contour $\partial\mc{B}$ indeed has the correct exceedence probability with respect to $\mc{E}$. See \cite{HusebyVE-EnvCont-ESREL2017} for further details regarding this.

\section{Buffered environmental contours}
\label{sec: buffered_environmental_contours}
In this section, we introduce a new concept called \emph{buffered environmental contours}. This combines the ideas behind buffered failure probabilities and environmental contours. Before we introduce the main results we review a result on superquantiles which will be essential in our approach (See \cite{Rockafellar}.)

\begin{prop}
\label{prop:monotonicity_of_superquantiles}
Let $g_1$ and $g_2$ be two performance functions such that $g_1(V) \leq g_2(V)$ almost surely, and let $\bar{q}_{1,\alpha}$ and $\bar{q}_{2,\alpha}$ denote the $\alpha$-superquantiles of $g_1$ and $g_2$ respectively. Then $\bar{q}_{1,\alpha} \leq \bar{q}_{2,\alpha}$.
\end{prop}
As a corollary of this result we get the following result on buffered failure probabilities:
\begin{cor}
\label{cor:monotonicity_of_bufProb}
Let $g_1$ and $g_2$ be two performance functions such that $g_1(V) \leq g_2(V)$ almost surely, and let $\bar{p}_{1,f}$ and $\bar{p}_{2,f}$ denote the buffered failure probabilities of $g_1$ and $g_2$ respectively. Then $\bar{p}_{1,f} \leq\bar{p}_{2,f}$.
\end{cor}

For a given performance function $g$ its failure probability, $p_f$, can be computed based on the failure region of $g$ alone. In contrast, computing the buffered failure probability, $\bar{p}_f$, requires more detailed information about the distribution of $g$. We indicate this by expressing $\bar{p}_f$ as a function of $g$ and denoted $\bar{p}_f(g)$.

\medskip

Just as for classical environmental contours, a \emph{buffered environmental contour} is the boundary $\partial \bar{\mc{B}}$ of some suitable set $\bar{\mc{B}} \subseteq \mb{R}^n$. We shall now describe how the set $\bar{\mc{B}}$ can be constructed. As in the previous section we let $\mc{U}$ be the set of all unit vectors in $\mb{R}^n$, and let $\bm{u} \in \mc{U}$. Moreover, we let $P_e$ be a given target probability, and let $C(\bm{u})$ be defined by \eqref{eq:Cfunction}. In order to introduce buffering, we let:
\begin{equation}
\label{eq:bufCfunction}
\bar{C}(\bm{u}) := \E[\bm{u}' \bm{V} | \bm{u}' \bm{V} > C(\bm{u})].
\end{equation}
Given the distribution of $\bm{V}$, the function $\bar{C}(\bm{u})$ can easily be estimated by using Monte Carlo simulation. As in \subsecref{sec: mc_contours}, we let $\bm{V}_1, \ldots, \bm{V}_N$ be a random sample from the distribution of $\bm{V}$, and choose $\bm{u} \in \mc{U}$. Based on the sorted values $Y_{(1)} \leq Y_{(2)} \leq \cdots \leq Y_{(N)}$ we first estimate $C(\bm{u})$ by $Y_{(k)}$ as previously explained. We then estimate $\bar{C}(\bm{u})$ by computing the average value of the sampled values which are greater than $Y_{(k)}$. Thus, we estimate $\bar{C}(\bm{u})$ by:
$$
\hat{\bar{C}}(\bm{u}) = \frac{1}{N - k} \sum_{r > k} Y_{(r)}.
$$

\medskip

For each $\bm{u} \in \mc{U}$, we also introduce the halfspaces:
\begin{align*}
\bar{\Pi}^-(\bm{u}) &= \{\bm{v} : \bm{u}' \bm{v} \leq \bar{C}(\bm{u})\},\\
\bar{\Pi}^+(\bm{u}) &= \{\bm{v} : \bm{u}' \bm{v} > \bar{C}(\bm{u})\},
\end{align*}
similar to what we did in the previous section. Finally, we define the buffered environmental contour as the boundary $\partial \bar{\mc{B}}$ of the \emph{convex set} set $\bar{\mc{B}}$ given by:
\begin{equation}
\label{eq:buf_env_contourSet}
\bar{\mc{B}} := \bigcap_{\bm{u} \in \mc{U}} \bar{\Pi}^-(\bm{u})
\end{equation}

We observe that by \eqref{eq:bufCfunction} we obviously have that $\bar{C}(\bm{u}) > C(\bm{u})$. By comparing \eqref{eq:env_contourSet} and \eqref{eq:buf_env_contourSet}, it is easy to see that this implies that:
$$
\mc{B} \subset \bar{\mc{B}}.
$$
Thus, given that the same target probability $P_e$ is used to construct both contours, the buffered environmental contour is more conservative than the classical environmental contour.

\medskip

The next step is to identify a family $\mc{G}$ of performance functions defined relative to the set $\mc{B}$ such that $\bar{p}_f(g) \leq P_e$ for all $g \in \mc{G}$. We recall that for the classical environmental contour we chose to let $\mc{E}$ be the family of all convex failure regions which do not intersect with the interior of $\mc{B}$. Thus, one might think that the natural counterpart for buffered environmental contours would be to let $\mc{G}$ be the family of performance functions with convex failure regions which do not intersect with the interior of $\bar{\mc{B}}$. In this case, however, we need more control over the distributions of the performance functions. In order to do so we choose $\bm{u} \in \mc{U}$ and introduce the performance function $\Gamma(\bm{u}, \cdot)$ given by:
$$
\Gamma(\bm{u}, \bm{V}) = \bm{u}' \bm{V} - \bar{C}(\bm{u})
$$
By \eqref{eq:bufCfunction} we have:
\begin{align*}
\E[\Gamma(\bm{u}, \bm{V}) &| \Gamma(\bm{u}, \bm{V}) > C(\bm{u}) - \bar{C}(\bm{u})] \\
&= \E[\bm{u}' \bm{V} | \bm{u}' \bm{V} > C(\bm{u})] - \bar{C}(\bm{u}) = 0.
\end{align*}
Moreover, by \eqref{eq:Cfunction} we have:
\begin{align*}
\bar{p}_f(\Gamma(\bm{u}, \cdot)) &= P(\Gamma(\bm{u}, \bm{V}) > C(\bm{u}) - \bar{C}(\bm{u})) \\
&= P(\bm{u}' \bm{V} > C(\bm{u})) = P_e
\end{align*}
Since the unit vector $\bm{u}$ was arbitrarily chosen, we conclude that the performance function $\Gamma(\bm{u}, \cdot)$ has the desired buffered failure probability $P_e$ for all $\bm{u} \in \mc{U}$. 

We will use these performance functions as a basis for constructing the family $\mc{G}$ where the $\Gamma(\bm{u}, \cdot)$-functions serve as \emph{maximal} elements in this family. Note that the $\Gamma(\bm{u}, \cdot)$-functions now play a similar role as the halfspaces $\Pi^+(\bm{u})$ played in the construction of the family $\mc{F}$. Thus, we let $\mc{G}$ be the family of all performance functions $g$ for which there exists a $\bm{u} \in \mc{U}$ such that $g(\bm{v}) \leq \Gamma(\bm{u}, \bm{v})$ for all $\bm{v} \in \mc{V}$. By the above discussion the following result is immediate:
\begin{thm}
\label{thm:targetBufProbFamily}
For all $g \in \mc{G}$ we have $\bar{p}_f(g) \leq P_e$.
\end{thm}
\begin{proof}
Assume that $g \in \mc{G}$. Then there exists a $\bm{u} \in \mc{U}$ such that $g(\bm{V}) \leq \Gamma(\bm{u}, \bm{V})$ almost surely. Hence, by \corref{cor:monotonicity_of_bufProb} and the above calculations we have:
$$
\bar{p}_f(g) \leq \bar{p}_f(\Gamma(\bm{u}, \cdot)) = P_e.
$$
\end{proof}

Having constructed both the set $\bar{\mc{B}}$ and the family $\mc{G}$ we are now ready to introduce the \emph{buffered exceedence probability} of $\bar{\mc{B}}$ with respect to $\mc{G}$ defined as:
\begin{equation}
\label{eq:bufExceedenceProb}
\bar{P}_{e}(\bar{\mc{B}}, \mc{G}) := \sup \{\bar{p}_f(g) : g \in \mc{G}\}.
\end{equation}

We note that by the definition of $\mc{G}$ it follows that $\Gamma(\bm{u}, \cdot) \in \mc{G}$ for all $\bm{u} \in \mc{U}$. Hence, we get:
\begin{align*}
\bar{P}_{e}(\bar{\mc{B}}, \mc{G}) 
&= \sup \{\bar{p}_f(g) : g \in \mc{G}\} \\[2mm]
&= \sup \{\bar{p}_f(\Gamma(\bm{u}, \cdot)) : \bm{u} \in \mc{U}\} = P_e,
\end{align*}
Thus, we conclude that the contour $\partial\bar{\mc{B}}$ indeed has the correct buffered exceedence probability with respect to $\mc{G}$.

\bigskip

If $g \in \mc{G}$ and $g(\bm{v}) \leq \Gamma(\bm{u}, \bm{v})$ for all $\bm{v} \in \mc{V}$, we have:
\begin{align*}
\mc{F}(g) &\subseteq \mc{F}(\Gamma(\bm{u}, \cdot)) \\
&= \{\bm{v} : \bm{u}' \bm{v} - \bar{C}(\bm{u}) > 0 \} \\
&= \{\bm{v} : \bm{u}' \bm{v} > \bar{C}(\bm{u}) \} = \bar{\Pi}^+(\bm{u})
\end{align*}
Thus, the failure region of a performance function $g \in \mc{G}$ does not overlap with the interior of the set $\bar{\mc{B}}$, but is contained within a halfspace supporting $\bar{\mc{B}}$. This is similar to the relation between failure regions in the family $\mc{E}$ and the set $\mc{B}$ for the classical environmental contours. However, as already pointed out, knowledge about the failure region of a performance function is not sufficient to ensure that the performance function has the correct buffered failure probability.

\bigskip

It may be argued that the choice of the $\Gamma(\bm{u}, \cdot)$-functions as maximal elements in the family $\mc{G}$ is too restrictive. In order to have a more flexible framework, it is possible to consider a slightly more general approach where we define:
\begin{equation}
\label{eq:bufCfunction}
\bar{C}_a(\bm{u}) := \E[a \bm{u}' \bm{V} | \bm{u}' \bm{V} > C(\bm{u})] = a \bar{C}(\bm{u}),
\end{equation}
where $a$ is a positive constant. By increasing the $a$-factor, the contour may be inflated so that it can be used for steeper performance factors. 

On the other hand it should be noted that to ensure that a given performance function $g$ has the correct buffered failure probability, it is not necessary that $g(\bm{v})$ is dominated by some $\Gamma(\bm{u}, \cdot)$-function for \emph{all} $\bm{v} \in \mc{V}$. It is sufficient that this holds for $\bm{v}$-values corresponding to the upper tail area of $g$.

\section{Numerical example}
\label{sec: numerical_example}
In this subsection we illustrate the proposed method by considering a numerical example introduced in \cite{VanemB-G:JOMAE2015}. More specifically, we consider joint long-term models for \emph{significant wave height}, denoted by $H$, and \emph{wave period} denoted by $T$. A marginal distribution is fitted to the data for significant wave height and a conditional model, conditioned on the value of significant wave height, is subsequently fitted to the wave period. The joint model is the product of these distribution functions:
$$
f_{T,H}(t, h) = f_H(h) f_{T|H}(t|h)
$$
Simultaneous distributions have been fitted to data assuming a three-parameter Weibull distribution for the significant wave height, $H$, and a lognormal conditional distribution for the wave period, $T$. The three-parameter Weibull distribution is parameterized by a location parameter, $\gamma$, a scale parameter $\alpha$, and a shape parameter $\beta$ as follows:
$$
f_H(h) = \frac{\beta}{\alpha} \left(\frac{h-\gamma}{\alpha}\right)^{\beta-1} e^{-[(h-\gamma)/\alpha]^{\beta}}, \quad h \geq \gamma.
$$
The lognormal distribution has two parameters, the log-mean $\mu$ and the log-standard deviation $\sigma$ and is expressed as:
$$
f_{T|H}(t|h) = \frac{1}{t \sqrt{2\pi}} e^{-[(\ln(t) - \mu)^2 / (2 \sigma^2)]}, \quad t \geq 0,
$$
where the dependence between $H$ and $T$ is modelled by letting the parameters $\mu$ and $\sigma$ be expressed in terms of $H$ as follows:
$$
\mu = E[\ln(T) | H = h] = a_1 + a_2 h^{a_3},
$$
$$
\sigma = SD[\ln(T) | H = h] = b_1 + b_2 e^{b_3 h}.
$$
The parameters $a_1, a_2, a_3, b_1, b_2, b_3$ are estimated using available data from the relevant geographical location. In the example considered here the parameters are fitted based on a data set from North West Australia. We consider data for two different cases: \emph{swell} and \emph{wind sea}. The parameters for the three-parameter Weibull distribution are listed in \tblref{tbl:weibull_parameters}, while the parameters for the conditional log-normal distribution are listed in \tblref{tbl:lognormal_parameters}. In all the examples we use a return period of 25 years. The models are fitted using sea states representing periods of 1 hour. Thus, we get 24 data points per 24 hours. Thus, the desired exceedence probability is given by:
$$
P_e = \frac{1}{25 \cdot 365.25 \cdot 24} = 4.5631 \cdot 10^{-6}.
$$
For more details about these examples we refer to \cite{VanemB-G:JOMAE2015}.

\medskip

\begin{table}[h!]
\begin{center}
\caption{Fitted parameter for the three-parameter Weibull distribution for signifcant wave heights}
\medskip
\small{
\begin{tabular}{lccc}
\hline
 & $\alpha$ & $\beta$ & $\gamma$ \\
\hline
Swell & 0.450 & 1.580 & 0.132 \\
Wind sea & 0.605 & 0.867 & 0.322 \\
\hline
\end{tabular}
}
\label{tbl:weibull_parameters}
\end{center}
\end{table}

\begin{table}[h!]
\begin{center}
\caption{Fitted parameter for the conditional log-normal distribution for wave periods}
\medskip
\small{
\begin{tabular}{llccc}
\hline
 & & $i = 1$ & $i = 2$ & $i = 3$ \\
\hline
Swell     & $a_i$ & 0.010 & 2.543 & 0.032 \\
		  & $b_i$ & 0.137 & 0.000 & 0.000 \\
Wind sea  & $a_i$ & 0.000 & 1.798 & 0.134 \\
		  & $b_i$ & 0.042 & 0.224 & -0.500 \\
\hline
\end{tabular}
}
\label{tbl:lognormal_parameters}
\end{center}
\end{table}

The classical environmental contours are estimated based on the methods presented in \cite{HusebyVN-EnvCont-OE2013}. More specifically, we have used Method 2 presented in this paper. The buffered environmental contours are estimated in exactly the same way, except that $\hat{C}(\bm{u})$ is replaced by $\hat{\bar{C}}(\bm{u})$ for all $\bm{u} \in \mc{U}$.

\medskip

In \figref{fig: bufContour_nwa_sw_25} and \figref{fig: bufContour_nwa_ws_25} the resulting environment contours are shown. As one expected, the classical environmental contours are located inside their respective buffered contours. Thus, since the target probability $P_e$ is the same for both types of contours, the buffered contours are more conservative than the classical contours.

\begin{figure}[h!]
\begin{center}
\includegraphics[width=0.5\textwidth]{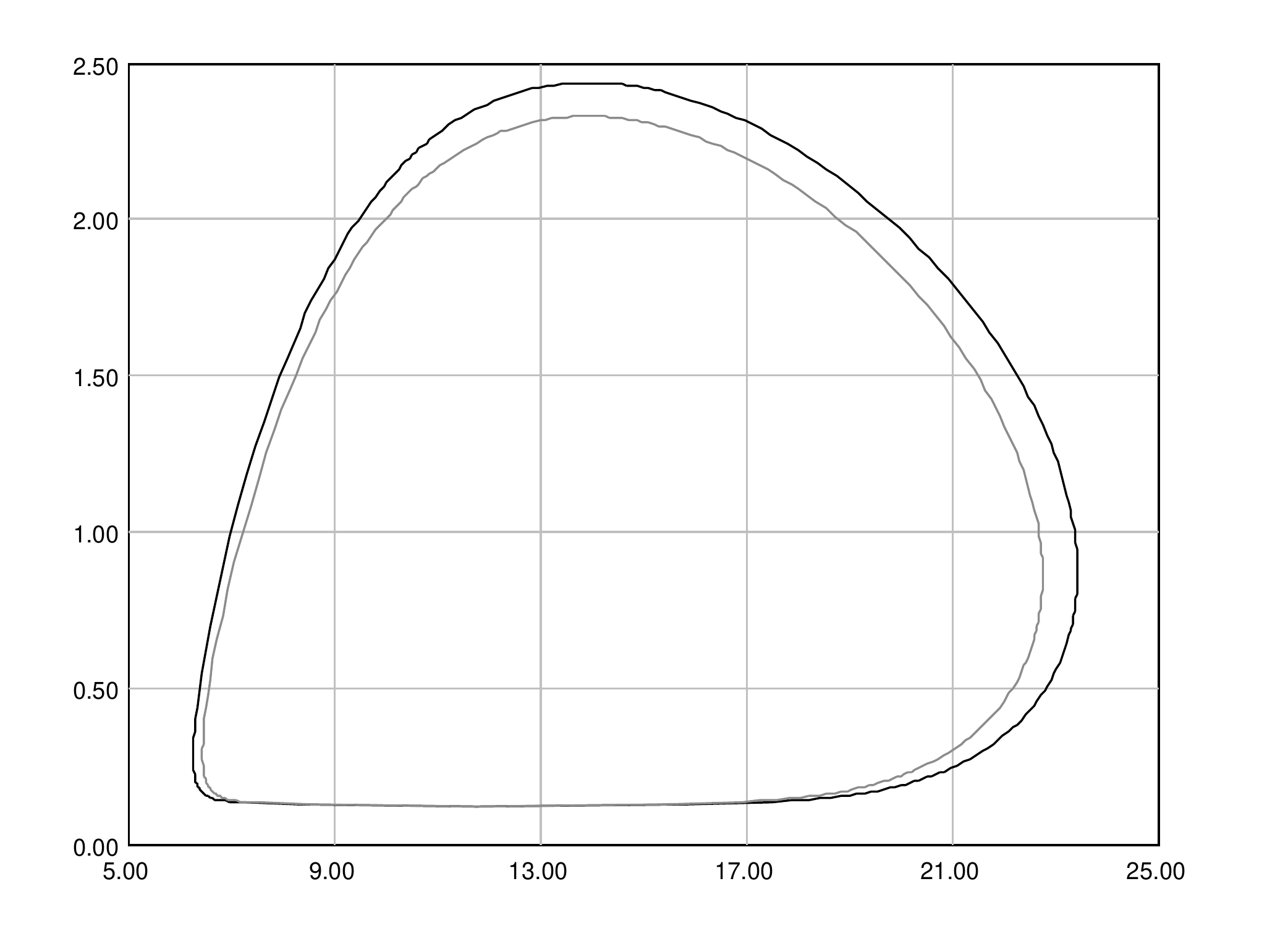}
\end{center}
\caption{Buffered environmental contour (black) and classical environmental contour (gray) for North West Australia Swell with return period 25 years.}
\label{fig: bufContour_nwa_sw_25}
\end{figure}
 
\begin{figure}[h!]
\begin{center}
\includegraphics[width=0.5\textwidth]{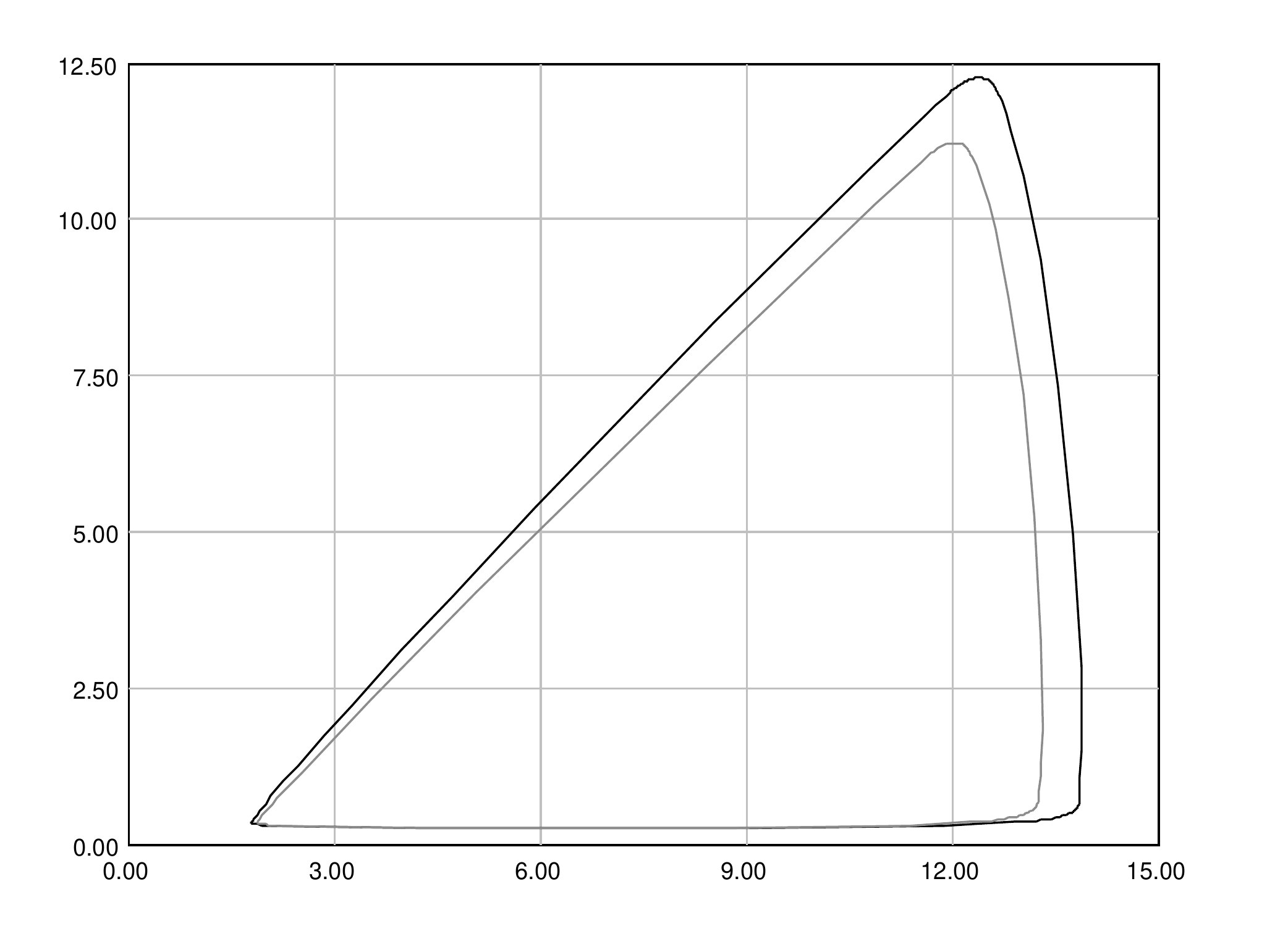}
\end{center}
\caption{Buffered environmental contour (black) and classical environmental contour (gray) for North West Australia Wind sea with return period 25 years.}
\label{fig: bufContour_nwa_ws_25}
\end{figure}

\section{Conclusions and future work}
\label{sec: conclusion}

In the present paper we have introduced the concept of buffered environmental contours, and shown how such contours can be estimated using Monte Carlo simulations. Such contours do not just take into account the probability of failure, but also the consequences of a failure. This is relevant e.g., when analysing the risk of flooding at a given location. While it may not be possible to prevent floodings from occurring, the damage caused by such an event can vary a lot depending on how much the water has risen above the normal level. In some cases only minor damages may be the result. In other cases the consequences can be catastrophic.

For a given target probability, $P_e$ buffered environmental contours are generally more conservative than the classical environmental contours. However,  in cases where the consequences are more important than the triggering event itself, a higher target probability might be acceptable as long as the damages are manageable. Thus, in real-life applications a buffered environmental contour may not be so conservative after all. At the same time these contours provide much more information about the tail area of the environmental variables. This may be very useful when a design is optimized.

The buffered environmental contours proposed in this paper are the natural extension of the Monte Carlo contours introduced in \cite{HusebyVN-EnvCont-OE2013}. In particular both contour types are boundaries of convex sets. Sometimes this restriction may lead to contours which include areas of very low probability. Thus, it would be of interest to investigate other ways of constructing buffered contours. In particular, it is possible to modify contours obtained by using the Rosenblatt transformation so that they include buffering. To make this work, however, evaluating the resulting contours becomes very important. The evaluation framework described in \cite{HusebyVE-EnvCont-ESREL2017} may serve as a starting point.

Future work in this area also includes the use of buffered environmental contours in design optimization, but with additional design constraints. The question is how such additional constraints can be dealt with. An initial idea is to apply a Lagrange duality method in order to transform the problem into a previously known form.

It would also be interesting to compare buffered environmental contours to the conservative environmental contours defined by \cite{Leira}. The contours defined in \cite{Leira} are typically larger sets than the environmental contours considered in \secref{sec: environmental_contours}, which means that they are more conservative when it comes to classifying structures as safe.

Another idea which requires further investigation is how time can be introduced into this model in a less restrictive way. As mentioned in \subsecref{sec: returnPeriods}, we consider average stochastic environmental conditions $\bm{V}_1, \bm{V}_2, \ldots$ over some specified time intervals and assume independence and identical distributions of the $\bm{V}_i's$. A more realistic approach would be to introduce a stochastic process in continuous time modelling the environmental situation. It is interesting to see how this affects the model and what consequences this has for the design optimization.

\section*{Acknowlegdements}
\label{sec:acknowledgements}
\small{
This paper has been written with support from the Research Council of Norway (RCN) through the project \emph{ECSADES} Environmental Contours for Safe Design of Ships and other marine structures.
}


\end{document}